\documentclass[final]{siamltex}

 \usepackage{mathrsfs}
 \usepackage{amssymb}
 \usepackage{graphicx}

\title{Analysis on the computability over the efficient utilization problem of the four-dimensional space-time\thanks{This work was supported by the National Natural Science Foundation of China (Grant No. 60773194) and the Fundamental Research Funds for the Central Universities (HUST: 2010MS099)}}

\author{Wenqi Huang
        \and Kun He \thanks{School of Computer Science and Technology, Huazhong University of Science and Technology, Wuhan 430074, China ({\tt Corresponding author. brooklet60@gmail.com})}}

\begin{document}

\maketitle

\begin{abstract}
This paper formally proposes a problem about the efficient utilization of the four dimensional space-time. Given a cuboid container, a finite number of rigid cuboid items, and the time length that each item should be continuous baked in the container, the problem asks to arrange the starting time for each item being placed into the container and to arrange the position and orientation for each item at each instant during its continuous baking period such that the total time length the container be utilized is as short as possible. Here all side dimensions of the container and of the items are positive real numbers arbitrarily given. Differs from the classical packing problems, the position and orientation of each item in the container could be changed over time. Therefore, according to above mathematical model, the four-dimensional space-time can be utilized more truly and more fully. This paper then proves that there exists an exact algorithm that could solve the problem by finite operations, so we say this problem is weak computable. Based on the understanding of this computability proof, it is expected to design effective approximate algorithms in the near future. A piggyback work completed is a strict proof on the weak computability over general and natural case of the three-dimensional cuboid packing decision problem that all parameters are positive real numbers.
\end{abstract}



\pagestyle{myheadings}
\thispagestyle{plain}
\markboth{WENQI HUANG AND KUN HE }{SIAM JOURNAL ON COMPUTING}

\section{Introduction}

In the real world, there are many NP-hard problems, including packing problem  \cite{ref4}, SAT problem \cite{ref41}, scheduling problem, etc.
  And many of them are related to the utilization of space and time, such as computer memory allocation, database storage allocation and cookies baking.
  At present, researchers usually regard time as a geometric dimension, present such problems as a typical NP-hard problem, the packing problem, and then design feasible yet not very efficient solutions.
  Relevant researches have been on the space scheduling problem \cite{ref1}, the assembly line problem \cite{ref2},  the multiprocessor scheduling problem \cite{ref4,ref3},
   the four-dimensional rectangular packing problem \cite{ref5}, and the high-dimensional packing problem with arbitrary-shaped items \cite{ref6}.

In fact, however, the items need to be baked for some time or
the data should be survival in a certain period could
change their locations or orientations
in their storage area, and such change could make more fully utilization on space and time.
Therefore, this paper will study the optimal scheduling of cuboids packing in a four-dimensional space-time, which has three dimensions of space and one dimension of time. We call it the four-dimensional space-time packing and scheduling problem for short.

Differs from the packing problem, here items could change their positions or orientations
after being placed into the container and before being taken out.
Hence this problem is a promotion and extension of the traditional packing problem.
 We can greatly enhance the flexibility of how to utilize the space and time, and improve the container's space-time utilization ratio by this new model.
Its applications include the manufacturing process in large-scale factories, the memory space allocation on computers, the parallel task scheduling on hypercubes \cite{ref7}, etc.
 In accordance with this new model, we presented an intuitive instance of baking biscuits \cite{ref8}, in which the entire working time of the oven was saved by a third.
 Based on this new mathematical model, it is expected to achieve a schedule that is both good and fast.
 In this way, relevant applications will be more rapid, economic and success. So, the potential economic and social benefits are noticeable.

A lot of valuable work has been done on the two- or three-dimensional packing problems, including the work of Beasley \cite{ref9}, Parreno \cite{ref10}, Bortfeldt \cite{ref11}, etc.
To the best of our knowledge, however, there is little work on the four-dimensional space-time packing and scheduling problem heretofore.
 Even for a simplified and preliminary version of this problem, the high-dimensional packing problem, there is only a handful of exploring work published in the literature.
The reason is that it is not easy to find a corresponding image for
 the four- and higher-dimensional packing problems in the three-dimensional space that we live in, and the computational complexity increases violently with the growth of the dimension.

Concentrating on the high-dimensional packing problems, Huang et al. \cite{ref1} proposed in 1991 that
the space scheduling problem could be regarded as a four-dimensional packing problem,
and a quasi-physical method could be obtained after a proper extension of the quasi-physical method \cite{ref1_1} for solving the three-dimensional packing problems.
Another early discussion on high-dimensional packing problems was presented by Barnes in 1995\cite{ref13}, in which he analyzed a case of packing $1 \times 1 \times ... \times 1 \times n$ rods.
Later on, Fekete et al. discussed a general case of the $n$-dimensional packing problem \cite{ref14,ref15,ref16}.
By using a graph-theoretical characterization of feasible packings, they proposed a  branch-and-bound framework with new classes of lower bounds,
and developed a two-level tree search algorithm for solving the high-dimensional packing problem to optimality.
Fekete et al. also reported their computational results on two- and three-dimensional benchmarks.
However, they did not calculate the four- or higher-dimensional case.
Also, there are only two- and three-dimensional benchmarks in the OR-Library\cite{ref17} and PackLib2 \cite{ref18}.
Fortunately, increasing work on high-dimensional packings could be found in the last few years.
 In 2009, Harren \cite{ref12} proposed two approximate schemes for the general $d$-dimensional hypercube packing problem. And in 2010, Li et al. \cite{ref5} proposed a greedy heuristic for the four-dimensional rectangular packing problem. Li et al. also generated several test instances for the four-dimensional case.

The two-dimensional rectangular packing problem has been proven to be strongly NP-hard
\cite{ref19}. This four-dimensional space-time packing and scheduling problem is developed basing on its extension,
 the three-dimensional cuboid packing problem \cite{ref20},which has a higher computational complexity. And we need to consider the participation of time, a continuous parameter, for the computation.
 Therefore, whether this problem is computable is a difficult problem.

Besides, we further consider a general and natural case where the parameters are all real numbers. In recent decades, some researchers have studied the computability and the computational complexity of real numbers and real functions \cite{ref23,ref24,ref25,ref26}.
If there is a universal machine that could represent real numbers and do unit arithmetic or logical operation on them in finite time\cite{ref25}, then, again, if we can prove that there exists a deterministic algorithm that could solve the problem by finite arithmetic or logical operations on the real parameters, we may say this problem is weak computable.

 Based on our previous work on the three-dimensional rectangular packing problem \cite{ref20}, we conducted a preliminary study on this four-dimensional space-time packing and scheduling problem in \cite{ref8}.
 In this paper, we will present a formal mathematical description for this problem, and prove that there exists a deterministic algorithm that could solve this problem to optimality by finite operations.

\section{Problem description and computability analysis} This section presents a formal description of the four-dimensional space-time packing and scheduling problem and the main idea of its weak computability proof.

\subsection{Problem description}
The four-dimensional space-time packing and scheduling problem can be defined as follows.
In a three-dimensional Euclidean space, given a cuboid container with fixed length $ L $, width $ W $ and height $ H $,
and  given $ n $  ($n \in N^+$) cuboid items with each item $i$ in size ($l_i$, $w_i$, $h_i$), the time length $ T_i $ for each item to be processed (e.g. to be baked) is given too
 (here all parameters are positive real numbers).
The problem asks to provide a scheduling scheme to determine the following variables:

\begin{romannum}
\item the starting time $ S_i $ for item $i$ being placed into the container
 ( then the ending time of being taken out is $ S_i $ + $ T_i $);

\item the position and orientation of item $i$ at every instant of the time interval [$ S_i $, $ S_i $ + $ T_i $),which is called the survival period of item $i$. An item is called survival if it is in the container at the current time.
\end{romannum}

 The goal is to minimize the time that the last item being taken out from the container, namely, to minimize the makespan

 \centerline { max($ S_1 $ + $ T_1 $, $ S_2 $ + $ T_2 $, ..., $ S_n $ + $ T_n $)
 - min($ S_1 $, $ S_2 $, ..., $ S_n $)}

At any time of the schedule, every survival item should be located orthogonally and without protruding from the container, and there is no overlapping between any two survival items.

\subsection{Main idea for the computability proof}

If there exists an optimal schedule for the four-dimensional space-time packing and scheduling problem,
then naturally there must have an order for the items to be baked successively in this schedule.
Therefore, by enumerating all permutations of the items,
and by proving that the optimal solution can be found by a greedy strategy for each baking order,
we can prove that the solution with the shortest makespan among all solutions for different permutations is optimal for the original problem.

We prove the weak computability of this problem in three steps.
 First, we design an exact algorithm A$_0$ for the cuboid packing decision problem P$_0$.
 Then, by using A$_0$ as a core subprocedure, we design an exact algorithm A$_1$ for problem P$_1$, the four-dimensional packing and scheduling problem with order constraints.
 Last, by using A$_1$ as a core subprocedure, we design an exact algorithm A$_2$ for problem P$_2$, the original four-dimensional packing and scheduling problem.

\section{Computability on problem P$_0$}

\subsection{Problem description}

The cuboid packing decision problem P$_0$ can be defined as follows.
In a three-dimensional Euclidean space, given a cuboid container with fixed length $ L $, width $ W $ and height $ H $,
and  given $ n $  ($n \in N^+$) cuboid items with each item $i$ in size ($l_i$, $w_i$, $h_i$)
 (here all parameters are positive real numbers).
The problem asks whether there is a feasible placement that could place all the items into the container.
If so, then output the detailed layout.
 A placement is called feasible if each item in the container is located orthogonally (constraint 1) and without protruding (constraint 2), and there is no overlapping between any two items (constraint 3).

Therefore, any $3(n +1)$ positive real numbers $(L,W,H,l_1,w_1,h_1,...,l_n,w_n,$ $h_n)$
can uniquely identify a general and natural cuboid packing decision problem.

 It is reported in many recent literatures that the cuboid packing decision problem is NP-hard.
However, they usually refer to a special case that all parameters are integers.
To the best of our knowledge, for the general and natural case that all parameters are real numbers, there is no computability proof formally published thus far.

\subsection{Conceptions}
Consider the container be embedded into a three-dimensional Cartesian reference frame,
 in such a way that the lower-left-near corner coincides with the origin and
  the upper-right-far corner coincides with point $(L,W,H)$.
Follows are conceptions will be used in the follow-up proof.

\begin{definition}
 Object $i$ $(i\in\{0,1,...,n\})$
 Object 0 is constructed by all the points on the wall of the container and outside the container. Let item $i(i\in\{1,2,...,n\})$ be object $i$.
\end{definition}

\begin{definition}
 Orientation $ O_i $. There are six possible orientations $1, 2, ... ,6$
 for object $i(i\in\{1,2,...,n\})$, with its dimensions on x-, y-, z-axes being $(l_i, w_i, h_i)$,$(l_i, h_i, w_i)$, $(w_i, l_i, h_i)$,$(w_i, h_i, l_i)$,$(h_i, l_i, w_i)$ or $(h_i, w_i, l_i)$, respectively.
\end{definition}

\begin{definition}
Configuration. Define
$(x_1,y_1,z_1, ... , x_i,y_i,z_i, ... , x_n,y_n,z_n,$ $O_1,...,$ $O_i,...,O_n)$
 as a configuration, with $(x_i,y_i,z_i)$ being the coordinate of the lower-left-near vertex of object $i$ and $O_i$ being its orientation $(i\in\{1,2,...,n\})$.
\end{definition}

 Owing to the limits of the orientation, any configuration satisfies the first constraint of problem P$_0$ .

\begin{definition}
  Valid configuration. At the current configuration, define $V_{ij} $ as the intersection volume between object $i$ and $j$. A configuration is called valid if $V \triangleq \sum\limits_{i,j=0,i<j}^{n}V_{ij}=0$,
 namely, there is no intersection between any two of the $n+1$ objects.
\end{definition}

 For any given configuration, the value of $V$ is a definite real number.
 So $V$ is a function of $(x_1,y_1,z_1,...,x_n,y_n,z_n,O_1,...,O_n)$,
 whose domain is $((-\infty,\infty)^{3n}$ $\{1,...,6\}^n)$ and range$\subset[0,\infty)$.
 Constraint 2 is satisfied when there is no intersection between object $0$ and any other object.
 Constraint 3 is satisfied when there is no intersection between any two objects in $\{1,2,...,n\}$.
 Therefore, a configuration whose $V=0$ is a valid configuration that satisfies all the three constraints of problem P$_0$.

 \subsection{Computability proof} If there exists a feasible layout for problem P$_0$, then the position and orientation of each item is determined for this layout.
 Therefore, if we can enumerate the six different orientations for each item, that is to enumerate the $6^n$ different orientation settings for the $n$ items, and answer whether there exist a feasible layout for each of the $6^n$ orientation settings, then we can exactly solve problem P$_0$.
  So in this subsection, we first discuss the computability over a degeneration of P$_0$. The  degenerated problem, namely P$_0'$, has an additional constraint that the orientation of each item is fixed beforehand.

\vspace{3ex}
\begin{lemma}
 \label{lemma1}
 If the partial difference quotient of a multi-variable function $f(x_1,x_2,$ $...,x_m)$
is bounded($\leq K$), then $f$ is a continuous function. Here the partial difference quotient $\Delta f / \Delta x_i$ is defined as

 \centerline {$(f(x_1,...,x_i{''},...,x_m)-f(x_1,...,x_i',...,x_m))/(x_i'' - x_i').$}
\end{lemma}

\begin{proof}

\quad $|f(x_1+\Delta x_1,x_2+\Delta x_2, ...,x_{m-1}+\Delta x_{m-1},x_m+\Delta x_m)$
$- f(x_1,x_2, ...,x_{m-1},x_m)|$

$=|f(x_1+\Delta x_1, x_2+\Delta x_2, ...,x_{m-1}+\Delta x_{m-1},x_m+\Delta x_m)$

\quad\quad $- f(x_1+\Delta x_1,x_2+\Delta x_2, ...,x_{m-1}+\Delta x_{m-1},x_m)$

\quad $+ f(x_1+\Delta x_1,x_2+\Delta x_2, ...,x_{m-1}+\Delta x_{m-1},x_m)$
$ - f(x_1+\Delta x_1, x_2+\Delta x_2,...,x_{m-1},x_m)$

 \quad $+...$

\quad$+ f(x_1+\Delta x_1, x_2,...,x_{m-1},x_m) - f(x_1, x_2,...,x_{m-1},x_m)|$

$\leq |f(x_1+\Delta x_1, x_2+\Delta x_2, ...,x_{m-1}+\Delta x_{m-1},x_m+\Delta x_m)$

\quad\quad $- f(x_1+\Delta x_1,x_2+\Delta x_2, ...,x_{m-1}+\Delta x_{m-1},x_m)|$

\quad $+ |f(x_1+\Delta x_1,x_2+\Delta x_2, ...,x_{m-1}+\Delta x_{m-1},x_m)$
$ - f(x_1+\Delta x_1, x_2+\Delta x_2,...,x_{m-1},x_m)|$

 \quad $+...$

\quad$+ |f(x_1+\Delta x_1, x_2,...,x_{m-1},x_m) - f(x_1, x_2,...,x_{m-1},x_m)|$

$\leq K \cdot |\Delta x_m| + K \cdot |\Delta x_{m-1}|+ ... + K \cdot |\Delta x_2| + K \cdot |\Delta x_1|$

$\leq mK \sqrt {\Delta x_m^2 +\Delta x_{m-1}^2+...+\Delta x_2^2 +\Delta x_1^2 }$

As $\sqrt {\Delta x_m^2 +\Delta x_{m-1}^2+...+\Delta x_2^2 +\Delta x_1^2 }$ is the distance between point $(x_1+\Delta x_1,x_2+\Delta x_2, ...,x_{m-1}+\Delta x_{m-1},x_m+\Delta x_m)$ and point $(x_1,x_2, ...,x_{m-1},x_m)$, the absolute difference between the functions of the two points is less than or equal to a constant multiplying the distance between the two points.

So, function $f$ is continuous. Lemma \ref{lemma1} is proved. \qquad\end{proof}

 \vspace{3ex}
 \begin{lemma}
 \label{lemma2}
 If the orientation of each item is fixed beforehand, then $V$ is a continuous function of $(x_1,y_1,z_1, ... , x_n,y_n,z_n)$, and it is everywhere continuous on its entire domain $(-\infty,\infty)^{3n}$.
\end{lemma}

\begin{proof}

Let us consider a two-dimensional degeneration of problem P$_0'$. Define
 $S_{ij} $ as the intersection area between object $i$ and $j$, and
 $S \triangleq \sum\limits_{i,j=0,i<j}^{n}S_{ij}$.
Now we prove that $S(x_1,y_1, ... , x_n,y_n)$ is everywhere continuous on its entire domain $(-\infty,\infty)^{2n}$.

Consider $S_{12}(x_1,y_1,x_2,y_2, ... , x_n,y_n)$.
Let item $2$ holds still and item $1$ only does the translation movement
 in the $x$ direction(Fig. \ref{fig1} illustrates two positions for item $1$ during its translation process).
That is to say, let the values of $y_1, x_2, y_2$ remain unchanged
and the value of $x_1$ varies within the scope of $(-\infty,\infty)$.

\begin{figure}[ht]
\begin{center} \includegraphics[width=260pt]{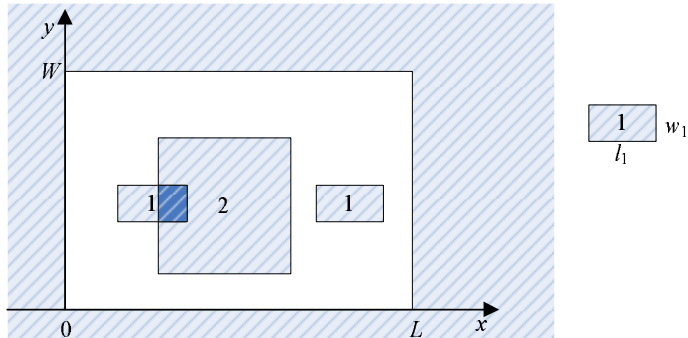} \end{center}
\caption{The relationship between $S_{12}$ and variable $x_1$.}
\label{fig1}
\end{figure}

As shown in Fig.~\ref{fig1}, the absolute value of the partial difference quotient

 \centerline {$|\Delta S_{12}/ \Delta x_1| = |\Delta S_{12}(x_1,$ $y_1,x_2,y_2,...,x_n,y_n) / \Delta x_1|$}

\noindent is everywhere less than or equal to $w_1$. Here $\Delta x_1 = x_1''-x_1'$ is a difference between any given real numbers $x_1''$ and $x_1'$ ($x_1'' \not= x_1'$). So, $|\Delta S_{12}/ \Delta x_1| \leq w_1$.

%$\therefore |\Delta S_{12}/ \Delta x_1| \leq w_1$.

$ \because w_1 \leq  max(l_1,w_1,l_1,w_2,...,l_n,w_n) \triangleq K$,

$\therefore$ for any $\Delta x_1, |\Delta S_{12}/ \Delta x_1| \leq K$.

Similarly,  $|\Delta S_{12}/ \Delta y_1| \leq K$,
 $|\Delta S_{12}/ \Delta x_2| \leq K$,$|\Delta S_{12}/ \Delta y_2| \leq K$.

$ \because x_3,y_3,...,x_n,y_n$ are all dummy variables for $S_{12}$,

$\therefore |\Delta S_{12}/ \Delta x_3|$,$|\Delta S_{12}/ \Delta y_3| $,...,
 $|\Delta S_{12}/ \Delta x_n|$,$|\Delta S_{12}/ \Delta y_n| $ are all equal to 0, and so are less than or equal to $K$.

Therefore, for any point at the $2n$-dimensional Euclidean
 space $(-\infty,\infty)^{2n}$,
 the absolute value of the partial difference quotient for $S_{12}(x_1,y_1,x_2,y_2,...,x_n,y_n)$ on any of its variables is less than or equal to $K$.

 Then, according to Lemma~\ref{lemma1}, function $S_{12}(x_1,y_1,x_2,y_2,...,x_n,y_n)$  is everywhere continuous on its entire domain $(-\infty,\infty)^{2n}$.

 Similarly, $S_{ij}(x_1,y_1,x_2,y_2,...,x_n,y_n)$  ($i,j \in \{0,1,...,n\}, i \not= j$) is everywhere continuous on its entire domain $(-\infty,\infty)^{2n}$ .

 So, $S(x_1,y_1, ... , x_n,y_n) \triangleq \sum\limits_{i,j=0,i<j}^{n}S_{ij}$ is everywhere continuous on its entire domain $(-\infty,\infty)^{2n}$ .

 Similarly, an analogical reasoning can prove that
 $V(x_1,y_1,z_1, ... , x_n,y_n,z_n)$ is everywhere continuous on its entire domain $(-\infty,\infty)^{3n}$.

  Lemma \ref{lemma2} is proved. \qquad\end{proof}

\vspace{3ex}
\begin{definition}
Feasible point set $\mathscr{H}$.
Suppose $(x_1,y_1,z_1, ... , x_n,y_n,z_n)$ is a point in a $3n$-dimensional Euclidean space.
If the orientation of each item is fixed beforehand,
then the feasible point set
 $\mathscr{H}$ $\triangleq \{(x_1,y_1,z_1, ... , x_n,y_n,z_n)|$ $ V(x_1,y_1,z_1, ... , x_n,$ $y_n,z_n)=0\}$.
\end{definition}

For simplicity, let's take a look at a degeneration of P$_0'$,
 the two-dimensional rectangular packing decision problem $(L,W,l_1,w_1,...,l_n,w_n)$.
 Suppose the side of $l_i$ is limited to the $x$ direction for each item,
 and $n=1$,$l_1=w_1=1$.
 If $L<1$ or $W<1$, then there is no solution for the problem, $\mathscr{H}=\varnothing$;
 If $L=1$ and $W=1$, then there is a unique solution, $\mathscr{H}=\{(0,0)\}$;
 If $L>1$ and $W>1$, then there are infinite solutions, the feasible point set
 $\mathscr{H}=\{(x,y)|0\leq x \leq L-1,0\leq y \leq W-1\}$.

Then, we have the following speculation and the corresponding proof.

\vspace{3ex}
\begin{theorem}
\label{th1}
  If there exists a feasible solution for problem P$_{0}'(L,W,H,l_1,$ $w_1,h_1,...,l_n,w_n,h_n,O_1,...,O_n)$, then the corresponding feasible point set
$\mathscr{H}$ is a nonempty, closed and bounded set in a $3n$-dimensional Euclidean space.
\end{theorem}

\begin{proof}

\begin{romannum}
\item Non-emptiness:
 According to the assumption that there exists a feasible solution to the problem,
$\mathscr{H}$ is not empty.

\item Boundedness:
 Because any item should be placed into the container completely,
 for any feasible point $(x_1,y_1,z_1, ... , x_n,y_n,z_n)$,
$\forall i\in\{1,2,...,n\}$,
$0\leq x_i \leq L,0\leq y_i \leq W,0\leq z_i \leq H$,
namely, $\mathscr{H}$  is bounded.

\item Closedness:
\end{romannum}

Let $p_1,p_2,...$ be a sequence of points in set $\mathscr{H}$,
and $\displaystyle{\lim_{i\rightarrow\infty}}p_i = p^*$.

 According to Lemma \ref{lemma2},
 $V$ is a continuous function of $(x_1,y_1,z_1, ... , x_n,y_n,z_n)$.

And according to a property of the continuous function,
the limit of the function sequence equals the function of the sequence limit
\cite{ref22}.

So,  $\displaystyle{\lim_{i\rightarrow\infty}}V(p_i)=
V(\displaystyle{\lim_{i\rightarrow\infty}}p_i) = V(p^*)$.$ \quad\quad\quad\quad\quad\quad\quad\quad\quad\quad\quad\quad\quad\quad\quad (1)$

$\because $ $\forall i\in \{1,2,...\}$,$p_i$ is a feasible point, namely $V(p_i)=0$,

$\therefore \displaystyle{\lim_{i\rightarrow\infty}}V(p_i)=0$,
$\quad\quad\quad\quad\quad\quad\quad\quad\quad\quad\quad\quad\quad\quad\quad\quad\quad\quad\quad\quad\quad\quad\quad (2)$

According to Eq. (1) and (2),

$V(p^*)=V(\displaystyle{\lim_{i\rightarrow\infty}}p_i)=
\displaystyle{\lim_{i\rightarrow\infty}}V(p_i)=0$.

$\therefore $ $p^*$ is a feasible point, $p^* \in \mathscr{H}$.

$\therefore $ Set $\mathscr{H}$ is closed in limit operation.

$\therefore $ $\mathscr{H}$ a closed set\cite{ref22}.

In summary, if there exists a feasible solution for problem P$_{0}$ under the constraint that the orientation of each item is fixed, then $\mathscr{H}$ is a nonempty, closed and bounded set in a $3n$-dimensional Euclidean space.

Theorem \ref{th1} is proved.\qquad\end{proof}

\vspace{3ex}
\begin{theorem}
\label{th2}
If there exists a feasible solution for problem P$_{0}'(L,W,H,l_1,$ $w_1,h_1,...,l_n,w_n,h_n,O_1,...,O_n)$,
 then there is a feasible point $p^*$ in $\mathscr{H}$,
 and in the corresponding layout of $p^*$, each item occupies a lower-left-near corner formed by other objects.
\end{theorem}

\begin{proof}

Define a potential energy function $U \triangleq \sum\limits_{i=1}^{n} (x_i+y_i+z_i)$,
whose domain is $\mathscr{H}$ and range$\subset[0,\infty)$. Note that $U$ is a continuous function.

As there must exist a global minimum point for a continuous function defined on a nonempty, closed and bounded set \cite{ref22},
 and according to Theorem \ref{th1}, so there must exist a global minimum point $p^*$ for $U$ defined on $\mathscr{H}$.

Consider a configuration $(x_1^*,y_1^*,z_1^*,...x_n^*,y_n^*,z_n^*)$ corresponding to point $p^*$, and consider $x_k^*,y_k^*,z_k^*$ for any item $k\in\{1,...,n\}$. At the current time, $x_k^*,y_k^*,z_k^*$ can no longer have smaller values independently and feasibly, or else it would be contradict with the conclusion that $p^*$ is a global minimum point for the potential energy $U$.
$x_k^*,y_k^*,z_k^*$ can no longer have smaller values independently and feasibly,
indicating that item $k$ has occupied a lower-left-near corner formed by other objects.

Theorem \ref{th2} is proved.\qquad\end{proof}

\vspace{3ex}
\begin{theorem}
\label{th3}
Problem P$_{0}'(L,W,H,l_1,w_1,h_1,...,l_n,w_n,h_n,O_1,...,O_n)$ is weak computable.
\end{theorem}

\begin{proof}

According to Theorem \ref{th2}, if there exists a feasible solution for the problem,
 then there is a feasible point $p^*$ in $\mathscr{H}$,
  and each item occupies a lower-left-near corner formed by other objects at the corresponding layout of $p^*$.

Thus for any item $i$ of this layout, it must paste with another item $i_1$ on its bottom,
 and $i_1$ must paste with another item $i_2$ on its bottom, ..., and $i_k$ must paste with the x-axis surface of object $0$ on its bottom. The relationships of $i$, $i_1,i_2, ..., i_k$ and object $0$ are as shown in Fig. \ref{fig2}.

$\therefore$ The value of coordinate $z_i$ for point $p^*$ is:

\quad $z_i = z_{i1}+(l_{i1}$ or $ w_{i1} $ or $ h_{i1})$

\quad \quad $= z_{i2}+(l_{i1}$ or $ w_{i1} $ or $ h_{i1})+(l_{i2}$ or $ w_{i2} $ or $ h_{i2})$

\quad \quad $=...$

\quad \quad $=(l_{i1}$ or $ w_{i1} $ or $ h_{i1})+...+ (l_{ik}$ or $ w_{ik} $ or $ h_{ik})$

$\because$ $\{i_1,i_2,...,i_k\} \in \{1,...,n\}$

$\therefore$ $z_i = (0$ or $1)(l_{1}$ or $ w_{1} $ or $ h_{1})$
                  $ + ... $
                  $ + (0$ or $1)(l_{n}$ or $ w_{n} $ or $ h_{n})$

\quad \quad $= (0$ or $l_{1}$ or $ w_{1} $ or $ h_{1})$
                  $ + ... $
                  $ + (0$ or $l_{n}$ or $ w_{n} $ or $ h_{n})$

$\therefore$ (the number for different $z_i$) $ \leq  4^n$  $(i = 1,2,...,n) $

\begin{figure}[ht]
\begin{center} \includegraphics[width=280pt]{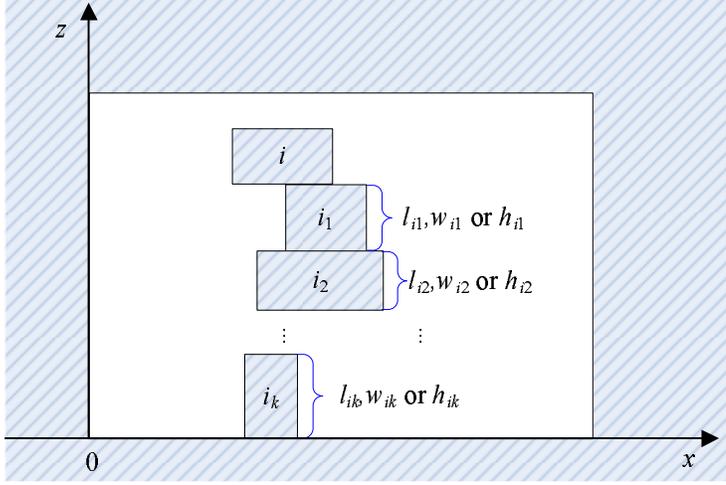} \end{center}
\caption{ The $z-$coordinate of item $i$ for the feasible point $p^*$.}
\label{fig2}
\end{figure}

Similarly, the number for different $x_i$ or $y_i$ both $ \leq  4^n$  $(i = 1,2,...,n) $

$\therefore$  The number for different $(x_1,y_1,z_1, ... , x_n,y_n,z_n)$
$ \leq   (4^n)^{3n} = 4^{3n^2}$ for solution $p^*$.

$\therefore$ By enumerating the $4^{3n^2}$ different layouts and by checking their feasibility one by one, we can exactly solve problem P$_0'$ by finite arithmetic or logical operations.

So, P$_0'$ is weak computable. Theorem \ref{th3} is proved.\qquad\end{proof}

\vspace{3ex}
\begin{theorem}
\label{th4}
The cuboid packing decision problem P$_{0}$ is weak computable.
\end{theorem}

\begin{proof}

Design an enumeration algorithm A$_0$ for the cuboid packing decision problem basing on the following strategies.

For problem P$_{0}(L,W,H,l_1,w_1,h_1,...,l_n,w_n,h_n)$,
there are a total of $n$ cuboid items, and there are six orientations for each item,
so there are a total of $6^n$ set of orientation configurations.

According to Theorem \ref{th3}, for each set of orientation configuration,
we can exactly solve the problem by enumerating the $4^{3n^2}$ different layouts.

Therefore, we can exactly solve problem P$_0$ in $6^n\cdot 4^{3n^2}$ iteration of computations.

So, P$_{0}$ is weak computable. Theorem \ref{th4} is proved.\qquad\end{proof}

 \section{Computability on problem P$_1$}
 In this section, we discuss problem P$_1$,
 the four-dimensional packing and scheduling problem with order constraints.
%After finding an optimal solution for P$_1$, we can then enumerate
% all possible orders to find an optimal solution for the original problem P$_2$.
 The container can be regarded as a cuboid oven.
 Problem P$_1$ is a degeneration of the original problem P$_2$,
 and P$_1$ has an additional constraint
 that the items should be baked in the ascending order of their serial number.
 In P$_1$, the starting time $S_i$ of item $i$ should be less than or equal to
 the starting time $S_{i+1}$ of item $i+1$ $(i \in \{1,...,n-1\})$.
 In this section, we first present a greedy scheduling algorithm A$_1$ for problem P$_1$,
 and then prove its optimality.

\subsection{The greedy scheduling algorithm}
Without loss of generality, we assume that any single item could be completely placed into the oven (larger items can be deleted from the baking sequence beforehand).
 For a P$_1$ instance,
 algorithm A$_1$ first finds a positive integer $k$ such that the first $k$ items can be all placed into the oven by algorithm A$_0 $ but the first $k+1$ items can not.
 Then it places the first $k$ items into the oven by A$_0$ and
  bake them continuously until at least one of the items completes its baking work.
 Next, take out the completed one or more items,
 and move the remaining items in the oven to the baking sequence with their left baking time
 and in the original order.
 Now the problem transforms into a new instance of problem P$_1$ again.
 Above procedure is call an iteration,
and such kind of iterations repeat until all the items complete their baking work.

\vspace{3ex}
 \textbf {Algorithm A$_1$.}

0.Initialization:

 \quad $C = (L,W,H)$,the cuboid oven;

 \quad $B = <(l_1,w_1,h_1),(l_2,w_2,h_2),...,(l_n,w_n,h_n)>$,the cuboid item sequence;

 \quad $F = <T_1,T_2,...,T_n>$,the remaining baking time;

 \quad $j = 1$,the time number;

 \quad $t_j = 0$,the initial time;

 \quad $R = \varnothing$,items already in the container at the current time;

 \quad $k = 0$,the maximum number of items already baked or are baking.

1. At time $t_j$ :%$\mathscr{T}_j$

1.1 Find the maximum $R$ set:

 \quad\quad do\{

\quad\quad\quad  $k=k+1$;

\quad\quad\quad  $R = R \cup \{$ the $k$-th item in $B \}$;

\quad\quad\quad  Pseudo place all the items in $R$ by algorithm A$_0$ (Note that there

\quad\quad\quad\quad  is no constraints for the placing order);

\quad\quad \}until the items in $R$ can no longer be all placed into the container

\quad\quad\quad  or $k = n$.

1.2 If the items in $R$ can not be all placed into the container by A$_0$ \{

\quad\quad \quad  $R = R - \{$ the $k$-th item in $B \}$;

\quad\quad \quad  $k=k-1$;

\quad\quad\quad Place the items in $R$ into the container by A$_0$ and bake them;

\quad\quad\quad The time length for the continuous bake  $ \Delta t_j \triangleq min\{T_i|i \in R\}$.

\quad\quad \}else \{ \quad //$k = n$,indicating that there is no items left

\quad\quad\quad Place the items in $R$ into the container by A$_0$ and bake them;

\quad\quad\quad Denote the latest time of the items being taken out by $t_{out}$;

\quad\quad\quad Exit.

\quad\quad \}

2. Preparations for time $t_{j+1}=t_j+\Delta t_j$:

\quad \ Update the remaining baking time for the items in $R$:

\quad\quad\quad For each $i \in R$,$T_i = T_i - \Delta t_j$;

\quad \ Update the $R$ set:

\quad\quad\quad Take out items whose remaining baking time is 0

\quad\quad\quad\quad from the container,and remove them from $R$;

\quad \ Update the time number: $j = j+1$;

\quad \ Return 1.

\vspace{3ex}

The corresponding greedy scheduling process is as shown in Fig. \ref{fig3}.
Here $t_j(j \in \{1,...,M\})$ corresponds to the entering time of the $j$-th batch of items,
 $t_j$ to $t_{j+1}(j \in \{1,...,M-1\})$ corresponds to beat $j$,
and the time length of this beat is $\Delta t_j$.
 At time $t_M$, all the remaining outside items have been placed into the oven, and we just keep taking out items that reach their baking time until all the items have completed their baking work.
 We represent the leaving time for the last takeout item by $t_{out}$, then the total time length for the whole baking process is $t_{out} - t_1= t_{out}$.

\begin{figure}[ht]
\begin{center} \includegraphics[width=380pt]{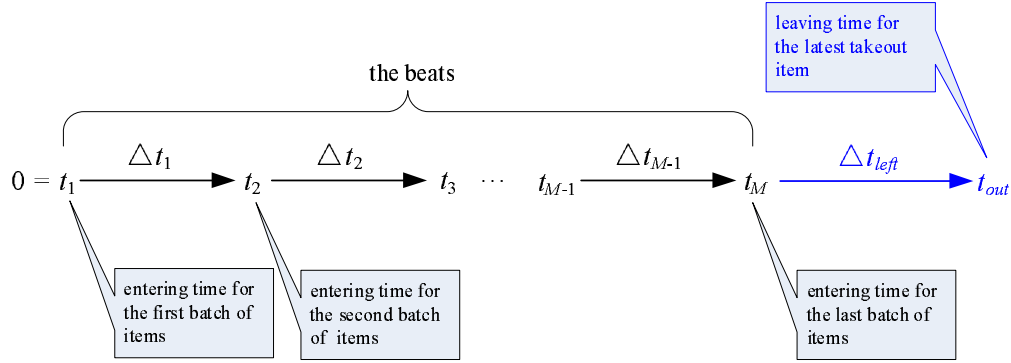} \end{center}
\caption{The greedy scheduling process.}
\label{fig3}
\end{figure}

\subsection{Computability proof}

This subsection combines two methods, the principle of mathematical induction and the reductio ad absurdum, to do the computability proof for problem P$_1$.

\vspace{3ex}
\begin{theorem}
\label{th5}
For problem P$_1$, suppose in the scheduling $\mathscr{S}$ achieved by algorithm A$_1$, the starting times of the items in the baking sequence are
$0=S_1 \leq S_2 \leq ... \leq S_{n-1} \leq S_n$,
then for any feasible scheduling $\mathscr{S'}$ that
$0 \leq S_1' \leq S_2' \leq ... \leq S_{n-1}' \leq S_n'$,
there must be $S_1 \leq S_1',S_2 \leq S_2',...,S_n \leq S_n'$.
\end{theorem}

\begin{proof}
The proof is mainly based on the principle of mathematical induction:
\begin{romannum}

\item For $k$=1, $S_1 = 0$,$S_1' \geq 0$, so $S_1 \leq S_1'$.

\item Suppose for any given $k$ $(k < n)$, $S_k \leq S_k'$, and we need to prove  $S_{k+1}\leq S_{k+1}'$. Now we prove it by contradiction and assume $S_{k+1} > S_{k+1}'$.
\end{romannum}

 \quad  $\because$ $S_1\leq ... \leq S_k \leq S_k' \leq S_{k+1}'$,$S_1'\leq ... \leq S_k' \leq S_{k+1}'$

 \quad  $\therefore$ For both $\mathscr{S}$ and $\mathscr{S'}$, any item in $\{1,2,...,k\}$ has been baked or is baking in the container at time $S_{k+1}'$.

 \quad  $\because$  $S_1 \leq S_1',S_2 \leq S_2',...,S_k \leq S_k'$

 \quad  $\therefore$  $S_1+T_1 \leq S_1'+T_1,S_2+T_2 \leq S_2'+T_2,...,S_k+T_k \leq S_k'+T_k$

 \quad  That is to say, for any item $i$ $(i \in \{1,...,k\})$, its leaving time in $\mathscr{S}$ is less than or equal to that in $\mathscr{S'}$ .

  \quad  $\therefore$ At time $S_{k+1}'$, the remaining items of $\{1,...,k\}$ in the container for $\mathscr{S}$ are also the remaining items of $\{1,...,k\}$ in the container for $\mathscr{S'}$.

  \quad  $\therefore$ At time $S_{k+1}'$, for items $\{1,...,k\}$, the remaining item set $U$ for $\mathscr{S}$, is a subset of the remaining item set $U'$ for $\mathscr{S'}$.

  \quad According to the assumption, at time $S_{k+1}'$, item $k+1$ is placed into the container in $\mathscr{S'}$, but it can not be placed into the container in $\mathscr{S}$.

 \quad It indicates that, at time $S_{k+1}'$, another schedule can place $U' \cup \{i+1\} $ into the container while A$_1$ can not place $U \cup \{i+1\}$ into the container,
 and $U \cup \{i+1\} \subset U' \cup \{i+1\}$.
 This is contradict with Theorem \ref{th4} that A$_0$ can solve problem P$_0$ to optimality by finite operations.

  \quad  $\therefore S_{k+1}\leq S_{k+1}'$.

  According to (1) and (2), $S_1 \leq S_1',S_2 \leq S_2',...,S_n \leq S_n'$.

 Theorem \ref{th5} is proved.\qquad\end{proof}

 \vspace{3ex}
\begin{theorem}
\label{th6}
The scheduling result achieved by algorithm A$_1$ is an optimal solution for problem P$ _1$.
\end{theorem}

\begin{proof}

For any feasible schedule $0 \leq S_1' \leq S_2' \leq ... \leq S_{n-1}' \leq S_n'$,
we may assume $S_1'=0$, or else we can delete $[0,S_1')$ to get a shorter schedule.

According to Theorem \ref{th5},
for the schedule $0=S_1 \leq S_2 \leq ... \leq S_{n-1} \leq S_n$ achieved by A$_1$ and
 any feasible schedule $0 \leq S_1' \leq S_2' \leq ... \leq S_{n-1}' \leq S_n'$,
there must be $S_1 \leq S_1',S_2 \leq S_2',...,S_n \leq S_n'$.

$\therefore$  $S_1+T_1 \leq S_1'+T_1,S_2+T_2 \leq S_2'+T_2,...,S_n +T_n\leq S_n'+T_n$.

 $\therefore$ max($S_1+T_1,S_2+T_2,...,S_n+T_n$)$\leq$ max($S_1'+T_1,S_2'+T_2, ...,S_n'+T_n$)

 $\because$ min($S_1$,$S_2$,...,$S_n$)=min($S_1'$,$S_2'$,...,$S_n'$)=0

$\therefore$ The makespan achieved by A$_1$ is less than or equal to that of any feasible schedule.

So the schedule achieved by A$_1$ is optimal. Theorem \ref{th6} is proved.\qquad\end{proof}

 \vspace{3ex}
\begin{theorem}
\label{th7}
The four-dimensional packing and scheduling problem with order constraints, problem P$_{1}$, is weak computable.
\end{theorem}

\begin{proof}

First, we prove that algorithm A$_1 $ can finish the computation in finite steps:

 The scheduling process of  A$_1$ is as shown in Fig. \ref{fig3}.
  At the end of each beat, at least one item will complete its baking work.
  As the number of items to be baked is finite,
 algorithm A$_1$ can reach $t_M$ in finite beats.

 For each time $t_j$ in $\{t_1, t_2 ,...,t_M\}$,
 A$_1$ has a maximum $R$ set and the corresponding layout
 by calling A$_0$ finitely.
 And according to Theorem \ref{th4}, A$_0$ can exactly solve P$_0$ by finite operations.
 So, the computation for each time $t_j$ can finish by finite operations.

 Therefore, algorithm A$_1 $ can reach $t_M$ by finite operations.

 At time $t_M$, the remaining items in the container is finite,
$\Delta t_{left}$ can be obtained by finite operations.

 In summary, algorithm A$_1$ can finish its computation by finite operations.

 Then according to Theorem \ref{th6},
 algorithm A$_1$ can exactly solve problem P$_1$ by finite operations.
  So, problem P$_1$ is weak computable.

 Theorem \ref{th7} is proved.\qquad\end{proof}

\section{Computability on the original problem P$_2$}
Based on above proofs, we can prove the weak computability of the original problem P$_2$ in this section.

 \vspace{3ex}
\begin{theorem}
\label{th8}
 The original four-dimensional packing and scheduling problem, problem P$_{2}$, is weak computable.
\end{theorem}

\begin{proof}

First, design algorithm A$_2$ as follows:

Enumerate all the permutations of the items.
And for each permutation, start from time 0 and solve problem P$_1$ by A$_1$. According to Theorem \ref{th7}, we can achieve an optimal schedule for the current permutation in finite operations.
Then, output the one with the minimum makespan among all solutions for different permutations.

Second, prove that the scheduling result obtained by A$_2$ is optimal:

 Let $Q$ be the scheduling result obtained by A$_2$,
 and let $Q'$ be any feasible schedule for P$_2$.

 As naturally there exists a baking sequence for $Q'$, and A$_2 $ has enumerated the corresponding sequence and has obtained an optimal schedule $Q_1$ for this sequence,
so the makespan of $Q_1$ $ \leq $ the makespan of $Q'$.

 As $Q$ has the minimum makespan among all permutations,
 so the makespan of $Q$$ \leq $ the makespan of $Q_1$.

 In summary,  the makespan of $Q$ $ \leq $  the makespan of $Q'$.

 Therefore, the scheduling result obtained by A$_2$ is optimal.

 As the number of all different permutations for finite items is finite,
 A$_2$ can complete the computation in finite operations
 and obtain an optimal schedule for problem P$_2$.

 Therefore, problem P$_2$ is weak computable. Theorem \ref{th8} is proved.\qquad\end{proof}

\section{Conclusion}

We present a four-dimensional packing and scheduling problem in this paper,
which has a significant scientific value and finds many practical applications.
It is neither a three-dimensional packing problem nor a four-dimensional packing problem,
but an a valuable expansion of the packing problem.

We further discuss a natural and general case of this problem that all parameters are real numbers, and prove its weak computability.
First, we prove that the cuboid packing decision problem with real parameters is weak computable.
Based on the proof, we then prove that the four-dimensional packing and scheduling problem is weak computable, indicating that there exists an exact and deterministic algorithm that could solve the problem by finite operations.
Although this algorithm has too high complexity to be practical,
it builds a firm footstone for the future work on approximate approaches.

\end{document}